\documentclass[runningheads]{llncs}
\usepackage{array}

\usepackage{tikz}
\usetikzlibrary{trees}
\usepackage[T1]{fontenc}

\usepackage{amsmath,amssymb}
\usepackage{graphicx}
\begin{document}
\title{ A Post-Quantum Secure Lattice-Based Forward-Secure Identity Based Encryption with Applications to Internet of Things Architecture}
\titlerunning{ }
 
%

\author{\large\bf Abhishek Kumar\inst{1} \and
 Vikas Srivastava \inst{2,*}  \and Sumit Kumar Debnath\inst{3}  \and
 Pantelimon St\u anic\u a \inst{4}  }
\authorrunning{A. Kumar et al.}
%
\institute{Department of Mathematics
National Institute of Technology Jamshedpur
Jamshedpur-831014, India\\
\email{ krabhishek1618@gmail.com}\\
\and Department of Mathematics National Institute of Technology Warangal, Hanamkonda, Telangana-506004
 \\
\email{ vsv@nitw.ac.in, vikas.math123@gmail.com}
\and Department of Mathematics
National Institute of Technology Jamshedpur
Jamshedpur-831014, India
 \\
\email{sdebnath.math@nitjsr.ac.in}
  \and
Department of Applied Mathematics
Naval Postgraduate School, Monterey, CA 93943, USA\\
\email{  pstanica@nps.edu}
}

\maketitle              
\begin{abstract}
The rapid expansion of the Internet of Things (IoT) has led to an unprecedented scale of data exchange across heterogeneous and resource-constrained devices. Ensuring confidentiality and secure key management in such environments is challenging. Traditional public-key infrastructures require heavy certificate-handling overhead. Identity-Based Encryption (IBE) offers a lightweight alternative by deriving public keys directly from device identities, making it attractive for IoT deployments. However, IoT devices are highly vulnerable to side-channel and key-extraction attacks, motivating the need for Forward-Secure IBE ({\sf FS-IBE}), where the compromise of a current secret key does not threaten past communications.

Existing {\sf FS-IBE} constructions based on classical hardness assumptions are not secure in the era of post-quantum, while the  lattice-based (LWE-based) forward-secure scheme suffer from large key and ciphertext sizes, limiting their suitability for constrained IoT systems. Here, we propose a new lattice-based fs-IBE scheme in the ring setting, relying on the RLWE assumption to achieve post-quantum security and significant efficiency gains. Our design uses trapdoor delegation with a minimal-cover mechanism over a binary tree. It results in compact public parameters and efficient per-epoch key updates. Compared to prior LWE-based constructions, our scheme reduces public key, secret key, and ciphertext sizes, and thus, making it better suited for practical IoT environments.

\keywords{Forward Security \and Ring-LWE \and Identity-Based Encryption \and IoT Security \and Lattice-Based Cryptography}
\end{abstract}

\section{Introduction}

The Internet of Things (IoT) represents a modern communication paradigm in which a wide range of physical objects such as devices, sensors, vehicles, and buildings are connected to the Internet. These objects are equipped with embedded electronics, software, and sensing modules that enable data collection and exchange \cite{zanella2014internet}. Many smart IoT applications rely on users’ personal or sensitive information to deliver customised and intelligent services \cite{nadikattu2018iot,srivastava2023multivariate}. As the number of  IoT devices increases, protecting the privacy and security of transmitted data has become a critical necessity \cite{mohanty2023quantum}.

In general, secure communication between IoT devices is achieved using cryptographic tools. Public-key techniques, including encryption and digital signatures, are commonly deployed to ensure confidentiality, data integrity, and authentication \cite{shafagh2014security,kelly2018optimisation}. Although these approaches provide strong security guarantees, they depend heavily on certificate management~\cite{hoglund2020pki4iot}. This introduces considerable computational and storage overhead, which is not suitable for IoT devices that are resource-constrained. To deal with this limitation,  an encryption scheme based on the identity of the user \cite{shamir1985identity} has been explored as an alternative for IoT settings. In IBE, the unique identity of a device, such as its ID number, acts as its public key. A device obtains its secret key from  Private Key Generator (PKG) \cite{zhang2024secure}. This eliminates the need for certificates and reduces management complexity \cite{sharma2022blockchain}. In cryptography, a scheme's security is ultimately dependent on the protection of secret keys. Studies have shown that side-channel attacks, which exploit power consumption patterns, electromagnetic leakage, or other physical characteristics of IoT devices, can be used to extract these keys~\cite{devi2020side}. Thus, while IBE removes certificate-related overhead, it is unable to protect against key-extraction attacks. Forward security provides an effective countermeasure in such scenarios. In a forward-secure design, if the current secret key is compromised still the previously shared encrypted data cannot be revealed~\cite{canetti2003forward,lu2011efficient}. This property ensures that past communication remains protected even if an adversary later obtains the private key \cite{yao2004id}. Another challenge arises from the fact that most traditional IBE schemes rely on classical hardness assumptions of problems like discrete logarithms or the factorisation of integers. With the help of quantum computers, these assumptions are expected to become insecure. Cryptography based on lattice has appeared as an effective solution, as lattice problems are believed to remain unbreakable even for quantum algorithms \cite{ajtai1996generating}. Motivated by this, Jin et al. in \cite{jin2024lattice} proposed a forward-secure IBE construction based hardness assumption of learning with error problem in lattice. However, their solution suffers from efficiency limitations, particularly due to large secret keys and ciphertexts. So, it is an urgent need to design an efficient post-quantum forward secure identity-based encryption scheme for IoTs.
 
 \subsection{Related Work} 
 The evolution of forward-secure and identity-based encryption schemes has relied on a wide range of cryptographic foundations, each offering distinct strengths and limitations. Shamir was the first who gave the idea of identity-based encryption in his work \cite{shamir1985identity}. Later in 2001, Boneh el al. proposed a fully functional IBE scheme in~\cite{boneh2001identity}. The scheme achieved security against chosen-ciphertext attacks in the random-oracle model based on the variant of the Computational Diffie–Hellman assumption. The construction was formulated using bilinear pairings over appropriate algebraic groups. The construction removed the need for certificate management, but remained vulnerable to quantum adversaries. In 2003, Canetti et al. introduced the concept of binary tree encryption and presented the first
non-interactive public-key encryption schemes with forward-security in~\cite{canetti2003forward}. The main construction was secure against chosen plaintext attack in the standard model under the decisional bilinear Diffie–Hellman assumption and the parameter sizes were designed to grow only logarithmically with the total epoch. Later, another scheme was formulated in the random-oracle model, which was more efficient and both designs were shown to be extendable to chosen-ciphertext security. But the schemes suffer the threat of quantum attacks along with the associated significant overhead of public key certificates. One year later, a scalable HIBE scheme with forward security was proposed by Yao et al. in~\cite{yao2004id}.  It was further shown that this construction can be adapted to obtain a public-key broadcast encryption scheme with forward security.  The fs-HIBE framework was also extended to support collusion-resistant multi-hierarchical identity-based encryption. The security of these schemes was analyzed under the bilinear Diffie–Hellman assumption within the random-oracle model. Despite these advancements, these schemes suffered the risk of quantum attack. Also, the high computational cost affected the scalability in resource-constrained environments like IoTs. Building on these advancements, Lu et al.~\cite{lu2011practical} proposed a practical FS-PKE scheme in 2011 that demonstrates security against selective-time-period and adaptive chosen-plaintext attacks in the standard model. It is shown that its performance parameter and total epoch are independent of each other, making the scheme more efficient than earlier FS-PKE constructions. It was further demonstrated how chosen-ciphertext security can be obtained in both the standard model as well as the random-oracle model.
But again, the scheme relies on hardness assumptions of classical hard problems which are not secure against quantum attack. Moreover, as a public-key encryption system, it inherits the traditional drawbacks of certificate-based infrastructures, including certificate generation, distribution, and revocation overhead. Singh et al. improved their own work \cite{singh2012lattice} by reducing the ciphertext size which is IND-sID-CPA (semantic) secure in the random-oracle model \cite{singh2013lattice}. Though these schemes were secure against quantum attacks but still faced the problem of high computational cost and large key size.
Another identity-based encryption from lattice with forward security and  additional property of keyword search was given by Yang et al.~\cite{yang2023fs}, where the key-management limitations were addressed by extending the FS-PEKS framework into an identity-based setting, leading to the formulation of a lattice-based FS-IBE with keyword search (FS-IBEKS) scheme. And it was proven to have indistinguishable selective-identity chosen-plaintext security model (IND-sID-CPA) within the random-oracle framework. To further strengthen the guarantees, an additional FS-IBEKS construction was proposed in the standard model, and its security was proven under adaptive-identity, chosen-plaintext attacks (IND-ID-CPA).

Hierarchical Identity-Based Encryption (HIBE) extends the concept of Identity-Based Encryption (IBE) by distributing key generation responsibilities across different hierarchical levels, thereby reducing the workload on the central private key generator. Due to this hierarchical structure, HIBE is well suited for large-scale organizations and distributed systems, as it also helps in limiting the impact of secret key exposure. A lattice-based HIBE construction with improved efficiency and security was proposed in \cite{jiang2020chosen}. The authors presented a CPA-secure HIBE scheme based on the Ring Learning With Errors (R-LWE) problem and further extended it to achieve adaptive CCA security under the hardness assumption of the Shortest Vector Problem (SVP). However, the scheme does not provide forward security, which is an important requirement in IoT applications where long-term key exposure can compromise previously transmitted data.

In 2024, Jin et al. \cite{jin2024lattice}   constructed an IBE scheme with forward security based on lattice for Internet of Things (IoTs), secure against quantum attacks. It is built on the hardness assumption of LWE. The scheme employs a mechanism of minimal-cover within a binary-tree structure, and its security is rigorously proven to possess forward-secure, selective-identity, chosen-plaintext attack model. It has also indistinguishable selective ID chosen ciphertext attack security. But the scheme is not efficient due to higher communication and storage overhead for IoT devices.    The trade-off between efficiency, security level, and post-quantum resilience in  these cryptographic constructions are summarized in Table~\ref{tab1}.

\vskip2em
 \begin{table*}[!htb]
\caption{{Techniques, advantages and limitations of  related works}} \label{tab1} 
\centering
\resizebox{0.99\textwidth}{!}{%
\begin{tabular}{|p{2.3cm}| p{4.3cm}| p{6.6cm}| p{4.6cm}| }  \hline
Scheme	               & Cryptographic Techniques             &  Advantages              & Drawbacks/Limitations \\ \hline \hline
 Boneh and Franklin   \cite{boneh2001identity} & *  Billinear Diffie-Hellman assumption \newline * Bilinear pairings on elliptic cureves \newline *Identity based encryption \newline *Hash function  &   *No requirment of public key certificates \newline * Secure against Chosen ciphertext attack in random oracle model     &  *  Not secure against quantum attack \\ \hline

 Canetti et al. \cite{canetti2003forward}. & *Public key encryption   \newline *Hash function &  * Forward secure against chosen ciphertext attack  \newline * Efficient key update \newline * Secure against Chosen plaintext attack & * Not secure against quantum attack   \newline *Certificate managment problem  \newline  \\ \hline

Yao et al. \cite{yao2004id} & * Bilinear Diffie-Hellman (BDH) assumption\newline * Tree based time evolution of private key\newline * Idetity-based broadcast encryption & * Forward secure against chosen ciphertext attack
    &   * Large ciphertext size and computational overhead  \newline * Not quantum resistance \\ \hline

Lu et al. \cite{lu2011practical} & *Learning with error hardness assumption \newline *Forward secure public key encryption  \newline *Gaussian sampling on lattices  & * Forward secure against selective-time period chosen ciphertext attack(FS-ST-CCA) \newline *
Forward secure against selective-time period chosen plaintext  attack (FS-ST-CPA) &  *Higher storage and computational cost \newline * Load of Certificate managment  \newline * Not secure against quantum  attack   \newline \\ \hline

Singh et al. \cite{singh2013lattice} & *  Hardness assumption of LWE \newline *Hash functions modeled as a random oracle \newline *Identity based encryption on lattices &   *Post-quantum security   \newline *Forward secure against chosen plaintext attack \newline * Shorter ciphertext size &  * Large key size and  \newline *High computational cost  \newline *Inefficient   \newline \\ \hline

Yang  et al. \cite{yang2023fs} & *Trapdoor generation and delegation algorithms \newline *Decisional learning with error \newline *Identity-based encryption on lattices    \newline & * No requirment of public key certificates \newline *Forward secure against chosen plaintext attack \newline *Post-quantum security &  *  High computational cost \newline * Large ciphertext and key sizes  \newline *Inefficient \\ \hline

Jin et al. \cite{jin2024lattice}. & *  Learning with error hardness assumption \newline *Identity based encryption \newline * Mechanism of minimal cover in context of binary tree  &  *Postquantum security  \newline * Secure against chosen ciphertext attack \newline * Secure against Chosen plaintext attack & *Large key and ciphertext size \newline *complexity of implementation \newline *Inefficient  \newline  \\ \hline

\end{tabular}
}
\end{table*}

\subsection{Our Contribution}
In this work, we put forward a new lattice-based fs-IBE scheme in the ring setting. We now summarize the important technical contribution of our work.
\begin{itemize}
  \item We design the \textit{first}, to the best of our knowledge, RLWE-based
  forward-secure IBE scheme (namely \textsf{RFS-IBE}) that combines a binary tree
  structure with the minimal-cover mechanism and trapdoor delegation over
  ideal lattices.  The construction supports per-epoch key updates for each
  identity while keeping the public parameters compact.

  \item We give a formal security analysis of \textsf{RFS-IBE} and prove that
  it achieves forward-secure selective-identity CPA security in the random
  oracle model under the decisional RLWE assumption via a reduction to a dual
  Regev-type encryption scheme.  We further outline how to extend the
  construction to obtain CCA security using standard transformations.

  \item We optimize the sizes of secret keys, public keys and ciphertexts by
  working in the ring setting instead of the standard LWE setting.  Compared
  with the lattice-based forward secure of IBE of Jin et al.~\cite{jin2024lattice}, our
  scheme achieves strictly smaller key and ciphertext sizes, which directly
  translates into lower communication and storage overhead for IoT devices.

  \item We illustrate how \textsf{RFS-IBE} can be integrated into a practical
  IoT architecture, with a hospital backend server acting as the key
  generation center and gateways / medical sensors playing the role of IBE
  users.  This case study shows that our construction can protect historical
  medical data against key-extraction attacks and side-channel leakage while
  remaining efficient enough for deployment on constrained devices.
\end{itemize}


\section{Preliminaries}
Here, we briefly summarise the preliminary information required to understand the proposed protocol. 
Lets begin with the notations used throughout the paper (refer Table \ref{tab2}). Then we will describe the concept of lattices, like the integer lattice and ideal lattice, along with their hard problems, specifically LWE and RLWE. We will understand the general formation of the forward secure identity-based encryption scheme along with its security model. We will also look at dual Regev type public key encryption based on the RLWE, followed by some important algorithms like {\sf ringGenTrap, ringExtBasis, ringRandBasis and ringSampleD}. We will also revisit the concept of binary tree and the trapdoor assignment based on it.

\subsection{Notations}
The notations that  are used throughout this paper are listed in the table \ref{tab2}.
\begin{table}
\caption{ Table of Notations }
     \centering
     \begin{tabular}{|c|c|}\hline
     $\lambda$ & Security parameter\\\hline
       $R_q$   & The ideal ring given by $R_q$ = $ \mathbb{  Z}_q[x]/\langle x^n + 1 \rangle$.\\\hline
        $\mathcal{B}$   & Probabilistic poly-time algorithm.\\\hline
     $y \leftarrow \mathcal{B}(x) $     &   Algorithm $\mathcal{B}$ take $x$ as input and outputs $\mathcal{B}(x)$.  \\\hline
        $\|\mathbf{s}_i\|$  & $\ell_2$-norm of the vector $\mathbf{s}_i$  from  set of vectors $S := \{\mathbf{s}_0, \ldots, \mathbf{s}_{k-1}\}$ \\\hline
       $ \|S\|$& The norm of the set $S$ defined as $\max_{i \in \{0, \ldots, k-1\}}$$\|\mathbf{s}_i\|$ \\\hline
         $\tilde{S}$ & The Gram--Schmidt orthogonalization of the vectors in this set. \\\hline
          $\|\tilde{S}\|$ & The Gram--Schmidt norm of $\tilde{S}$ \\\hline
          $\chi$ & The discrete Gaussian distribution   \\\hline
     \end{tabular}
     \label{tab2}
 \end{table}
   
\subsection{Lattices background}  
A lattice, denoted by $ \Lambda$, in $\mathbb{R}^n$ is the set of all integer linear combinations of linearly independent basis vectors $B=\{\mathfrak{b}_1,\mathfrak{b}_2,\ldots,\mathfrak{b}_m\}$:
\[
\Lambda(B) = \left\{ \sum_{i=1}^{m} z_i \mathfrak{b}_i \mid z_i \in \mathbb{Z} \right\}.
\] 
Below, we describe the three families of integer lattices that arise naturally in cryptography.  
Let $q$ be a positive integer modulus, $A \in \mathbb{Z}_q^{n \times m}$ a matrix, and $\mathbf{u} \in \mathbb{Z}_q^n$ a vector.  
The three commonly used integer lattices are defined as follows:  

\[
\Lambda_q(A^T) = \Big\{ \mathbf{z} \in \mathbb{Z}^m \;\Big|\; \exists\, \mathbf{s} \in \mathbb{Z}_q^n 
\ \text{with}\ A^T \mathbf{s} \equiv \mathbf{z} \pmod{q} \Big\},
\]

\[
\Lambda_q^{\perp}(A) = \Big\{ \mathbf{z} \in \mathbb{Z}^m \;\Big|\; A \mathbf{z} \equiv \mathbf{0} \pmod{q} \Big\},
\]

\[
\Lambda_{q,\mathbf{u}}(A) = \Big\{ \mathbf{z} \in \mathbb{Z}^m \;\Big|\; A \mathbf{z} \equiv \mathbf{u} \pmod{q} \Big\}.
\]
  
{\it Discrete Gaussian:} Let $L$ be a subset of $\mathbb{Z}^{n}$, $c\in \mathbb{R}^{n}$ and $r$ be positive real number. We define
\[
    \rho_{r,c}(x) = exp\left(-\pi \frac{\|x-c\|^{2}}{r^{2}}\right)
\]
and for a set $L$, we let   
\[
\rho_{r, c}(L)= \sum_{x \in L}\rho_{r, c}(x). 
\]
Taking $c$ and $r$  as center and parameter, the discrete Gaussian distribution over $R$ can be defined in the following way: for all $x \in L$
\[
  \mathfrak{D}_{L, r, c}(x)= \frac{\rho_{r,c}(x)}{\rho_{r,c}(L)}.
\]

 Today, the hardness of several computational problems over lattices forms the backbone of post-quantum security. In the context of the proposed design, we are mainly concerned with the Ring Learning with Errors and its variants. 

LWE allows hardness reduction from the worst-case lattice problem, which is hard even for quantum computers. Unfortunately, LWE needs a larger computational time and key size. So, to improve efficiency, we can replace it with Ring LWE and its variants, which supports the construction of post-quantum secure cryptography.

\subsection{ Ring LWE}

For prime $q$ and $n$  as an integer power of 2 the polynomial ring can be defined as 
$ R_q = \mathbb{Z}_q[x]/\langle x^n + 1 \rangle$, where  and $\mathbb{Z}_q[x]$ is the set of polynomials over  $\mathbb{Z}_q$. Each block $n \times n$ of the corresponding matrix $A$  can be compactly expressed as a ring element
$a_i \in R_q$.

In the context of polynomial rings, lattices can be defined similarly.  
Let  $\mathbf{a} \in R_q^k$ a vector, and $u \in R_q$.  
Then the corresponding lattices are given by:  

\[
\Lambda_q(\mathbf{a}) = \Big\{ \mathbf{z} \in R^k \;\Big|\; \exists\, s \in R_q 
\ \text{such that}\ \mathbf{a}s \equiv \mathbf{z} \pmod{q} \Big\},
\]

\[
\Lambda_q^{\perp}(\mathbf{a}^T) = \Big\{ \mathbf{z} \in R^k \;\Big|\; \mathbf{a}^T \mathbf{z} \equiv 0 \pmod{q} \Big\},
\]

\[
\Lambda_{q,u}(\mathbf{a}^T) = \Big\{ \mathbf{z} \in R^k \;\Big|\; \mathbf{a}^T \mathbf{z} \equiv u \pmod{q} \Big\}.
\]
\\
{\it Ideal Lattice:} Consider an ideal $\mathfrak{I}$ of quotient ring $R= \mathbb{Z}[x]/\langle x^n + 1 \rangle$ then any sublattice of $\mathbb{Z}^{n}$ that corresponds to the ideal $\mathfrak{I}$ is called an ideal lattice. It can be proved easily that the quotient ring $R$ is isomorphic to the ring~$\mathbb{Z}^{n}.$

\noindent {\it Ring-LWE}. To describe the LWE in the ring setting, the matrix $A$ is required to satisfy an additional algebraic structure.  
Each column of $A$ can be interpreted as the coefficient vector of a polynomial $p(x)$ of degree at most $n-1$. Moreover, for some polynomial $f(x)$ of degree $n$, the vector of coefficients corresponding to $x \cdot p(x) \bmod f(x)$ must also appear as one of the columns of $A$.  If we assume $m = nk$ for an integer $k$, then the $i$-th $n \times n$ block of $A$ can be identified with a ring element $a_i \in R_q$.  
In this setting, the secret vector $\mathbf{s}$ is represented by an element $s \in R$, the error vector $\mathbf{e}$ corresponds to $(e_1, \ldots, e_k)^T \in R^k$, and the output vector $\mathbf{b}$ is expressed as $(b_1, \ldots, b_k)^T \in R_q^k$~\cite{lyubashevsky2013toolkit}.  

The multiplication of a block of $A$ with $\mathbf{s}$ is then equivalent to multiplication of the corresponding ring elements $a_i$ and $s$.  Hence, the \emph{search version} of the ring-LWE problem is to determine the secret $s$ given the set of samples  
\[
\{(a_i, b_i = a_i s + e_i)\}_{i=1}^k.
\]  
The \emph{decision version} instead asks to distinguish such samples from a uniformly random one~\cite{lai2014trapdoors}. 

Due to this structure, cryptosystems built on ring-LWE  
assumptions achieve  higher efficiency compared to general
lattice-based constructions:
\begin{enumerate}
    \item Matrices of size $n \times n$ can be compressed into single 
    ring elements, reducing parameter sizes by about a factor of $n$.  
    \item The algebraic structure of \emph{ideal lattices} allows for 
    faster and more efficient computations.
\end{enumerate}

\begin{theorem}[\textup{\cite{lyubashevsky2010ideal}}]
Assuming that no polynomial-time quantum algorithm can approximate the Shortest Vector Problem (SVP) on ideal lattices in
$R$ within any polynomial factor, then any polynomially bounded collection of ring-LWE samples is indistinguishable from random for every polynomial-time adversary, including quantum ones.
\end{theorem}

 \subsection{General Construction of Forward-Secure Identity-Based Encryption}
\noindent
A forward-secure identity-based encryption ({\sf FS-IBE}) scheme   is described by a collection of five algorithms. $\mathcal{C}$ and $\mathcal{M}$  denote the  ciphertext space and message space respectively \cite{yao2004id}:

\begin{itemize}
    \item $(\mathit{params}, \mathit{msk}) \leftarrow {\sf FS-IBE.Setup}(\sf \lambda, N)$. The Private Key Generator (PKG) runs the setup procedure by taking the input $\lambda$ and $N$ as  a security parameter  and the total  epoch in the forward-secure system. It generates the public parameters $\mathsf{params}$  and the master secret key $msk$.
    
    \item $\mathit{sk_{id,0}} \leftarrow {\sf FS-IBE.KeyGen}(\mathit{params}, \mathit{msk})$.
    Given public parameters $\mathit{params}$, master secret key $\mathit{msk}$ and a user identity $id$, the algorithm outputs the initial secret key $sk_{id,0}$ corresponding to identity $id$.
    
    \item $\mathit{sk_{id,i+1}} \leftarrow {\sf FS-IBE.Update(\it params,  \mathit{sk_{id,i}})}$. For $params$ and a key $sk_{id,i}$ associated with the pair $(id,i)$, this procedure outputs the next secret key $sk_{id,i+1}$ for interval $(i+1)$, and securely removes $sk_{id,i}$.
    
    \item $\mathit{C}\leftarrow{\mathsf{FS-IBE.Encrypt}(\mathit{params}, \it id, M)}$. On input the public parameters $\it {params}$, the identity $id$ and an intervel pair $(id, t)$ and message $ M$, where $M \in \mathcal{M}$, the algorithm produces a ciphertext $C \in \mathcal{C}$.
    
    \item $({\mathit{M} or \perp}) \leftarrow {\sf FS-IBE.Decrypt(\mathit{params, C, sk_{id,i}}):}$ On input $\it {params}$, ciphertext $C$, secret key $sk_{id,i}$, with $C \in \mathcal{C}$, the algorithm recovers the message $M \in \mathcal{M}$ or outputs $\perp$ if decryption fails.
\end{itemize}

\noindent{\it Correctness}. If $sk_{id,i}$ is generated for the pair $(id,i)$ by either {\sf FS-IBE.KeyGen} or {\sf FS-IBE.Update}, then for any $M \in \mathcal{M}$, we have
\[
   {\sf FS\text{-}IBE.Decrypt}(\mathit{params}, \mathsf{FS\text{-}IBE.Encrypt}(\mathsf{params}, (id,i), M), sk_{id,i}) = M.
\] 

\subsection{Security Model}

The security notion for forward-secure IBE follows the framework of \cite{yao2004id}.  
We say that an FS-IBE scheme achieves \textit{forward-secure selective-identity chosen ciphertext security} ({\sf fs-sID-CCA}) if there is no adversary of probabilistic time which can succeed in the following game with more than negligible advantage.

\begin{description}
\item[{\sf Initialization}:] The adversary first commits to a challenge identity $id^{\ast}$. The challenger executes the \textsf{FS-IBE.Setup} algorithm and sends the public parameters ${\it params}$ to the adversary.
    
    \item[{\sf Phase 1 (Query Stage)}:]  For $i = 1, \ldots, m$, the adversary will make  $q_i$ queries and the  challenger react as: 
    \begin{itemize}
        \item {\sf Extraction query} $(id, j)$: The challenger generates the corresponding secret key $sk_{id,j}$ using \textsf{FS-IBE.KeyGen} and \textsf{FS-IBE.Update}, and returns it to adversary.
        \item {\sf Decryption query} $(id, C, j)$: The secret key, $sk_{id,j}$ is computed by the challenger using \textsf{FS-IBE.KeyGen} and \textsf{FS-IBE.Update}. It recovers the message  $M$ corresponding to the ciphertext $C$ with the help of the algorithm {\sf FS-IBE.Decrypt} and provides either the recovered message $M$ or the rejection symbol $\perp$.
    \end{itemize}
    
    \item[{\sf Challenge}:] Once {\sf Phase~1} concludes, messages $M_0, M_1 \in \mathcal{M}$ of equal length,  the challenge identity $id^{\ast}$ are submitted by the adversary along with target time period $t^{\ast}$.  
    The challenger computes $C^{\ast}$ using \textsf{FS-IBE.Encrypt}$(id^{\ast}, t^{\ast}, M_b)$ where $b$ is a ramdomly selected bit from \{0,1\}.
    and gives $C^{\ast}$ to the adversary.  
    One restriction applies: the adversary cannot request secret keys for $(id^{\ast}, t)$ with $t \leq t^{\ast}$.
    
    \item {\sf Phase 2 (Query Stage)}: For $i = m+1, \ldots, n$, adversary continues to make queries~$q_i$. The challenger reacts as in Phase~1, with the following rules:
    \begin{itemize}
        \item Extraction queries remain subject to the same restriction from the challenge step.
        \item Decryption queries $(id, C, j)$ are answered unless $(id, C, j) = (id^{\ast}, C^{\ast}, t^{\ast})$.
    \end{itemize}
    
    \item {\sf Guess}: The adversary finally give a bit $b' \in \{0,1\}$ as output. If $b'$ is same as $b$, the adversary is considered successful.
\end{description}
We can formulate adversary's advantage as:
\[
   \mathsf{Adv}_{\mathcal{A}, \mathcal{E}} = \big| \Pr[b = b'] - \tfrac{1}{2} \big|,
\]
where $\mathcal{A}$ is the adversary and $\mathcal{E}$ is the encryption scheme.

The above experiment models {\sf fs-sID-CCA} security.  If decryption queries are disallowed in the query phases, the resulting notion corresponds to {\sf fs-sID-CPA} security.

\subsection{RLWE-Based  PKE}
\label{sec.2.6}

{We now describe a public key encryption scheme that consists of the algorithms Key Generation, Encryption and Decryption. Its  depends on the hardness assumption of the RLWE problem. The scheme is shown to be CPA-secure. }\\

\noindent\textbf{Key Generation:}  
The secret key is selected as a short polynomial $s \sim \chi$.  
 $(\vec{a}, y)$ is the public key, where $\vec{a} \in R_q$ is chosen uniformly at random and  
\[
y = \vec{a} \cdot s + e \pmod{qR},
\]  
with $e \sim \chi$ being the error polynomial.

\medskip
\noindent\textbf{Encryption:}  
 For the encryption of the message $m \in R_{\{0,1\}}$, i.e., a polynomial with coefficients in $\{0,1\}$, sample random polynomials $r, x, x' \sim \chi$,
 we compute
\[
a = \vec{a}r + x \pmod{q}, \qquad
b = yr + x' + \Big\lfloor \tfrac{q}{2} \Big\rfloor m \pmod{q}.
\]

\medskip
\noindent\textbf{Decryption:}  
For  decryption, we compute  
\[
b - a \cdot s \pmod{qR} 
  = \Big\lfloor \tfrac{q}{2} \Big\rfloor m + er + x' - xs \pmod{q}.
\]  
By applying coefficient-wise rounding, the original message polynomial $m$ can be recovered. The scheme is correct provided that the error term $er + x' - xs$ is bounded by $q/4$.  
This condition allows the error size to be chosen on the order of $\sqrt{q}/2$.  The toughness of the decisional version of Ring-LWE problem with short secrets is used in two steps:  
first, to argue that uniformly random ring elements can replace the public key and second, to demonstrate that the ciphertext component $b$ is indistinguishable from uniform.  This scheme is notably efficient, achieving security that is believed to be of order of $2^{\Omega(n)}$,  
while all operations can be carried out in not more than $n \cdot \mathrm{polylog}(n,q)$ time.  
By selecting $q = \mathrm{poly}(n)$, we obtain a PKE scheme where  key generation,  
encryption, and decryption are all computable in quasilinear time with respect to the security parameter.  Furthermore, Lyubashevsky et al.~\cite{lyubashevsky2010ideal} demonstrated that breaking this scheme 
is at least as hard as solving certain worst-case problems on ideal lattices, which are conjectured  
to require $2^{\Omega(n)}$ time.

\subsection{Sampling Algorithms}

     Let $\vec{a}:= [ \vec{a}_{1}, \vec{a}_{2}, \dots, \vec{a}_{p}] \in R_{q}^{(l+pk)}$ where $\vec{a_{1}} \in R_{q}^{(k)}$ and $\vec{a_{i}} \in R_{q}^{(k)}$ for $i \in \{2,3, \dots p\}$. For $S \subseteq [p] \ 
([p]:=\{ 1, 2, \dots ,p\} ),  S= \{ i_{1}, i_{2}, \dots, i_{j} \} $, we set $\vec{a}_{s}:= [a_{{i}_{1}}, \dots \vec{a}_{{i}_{j}}]$, i.e.  we select the components of $\vec{a}$ according to S, when we treat  $\vec{a}_{{i}_{u}} \in R_{q}^{(l+k)} (u= 1,2 \dots j)$   as a component of $\vec{a}$. Without loss of generality, we let $S= [s]$, for some $s \in [p] $.
    \begin{lemma}[\textup{\cite{lai2014trapdoors}}]
    \label{lemma1}
    The PPT algorithm {\sf ringGenSamplePre} takes a ring element $\vec{a} \in R_{q}^{(l+k)}$, a set $S \subseteq [p]$ and a random basis $B_{s}$ of $\Lambda_{q}^{\perp}(\vec{a}_{s})$, a vector $y \in R_{q}^{l}$ and a positive integer $r \geq \| \tilde{B_{s} }\| .\omega(\sqrt{\log pk})$ and outputs $e \leftarrow {\sf ringGenSamplePre}(\vec{a}, B_{s}, S, y, r)$ that are statistically close to sampling a random error vector from the distribution $D_{\Lambda_{q}^{y}(\vec{a}), r}$.
\end{lemma}

\subsection*{New Trapdoor Generation} 
We describe the trapdoor generation in a ring lattice. We consider a distribution $\chi^{l \times k}$ over $R^{l \times k}$ and   $h \in R_{q}$ which is a nonzero element. The following lemma gives us a way to generate a trapdoor.  
\begin{lemma}[\textup{\cite{lai2014trapdoors}}]
    There exists an algorithm {\sf ringGenTrap} that takes a vector of ring element $\vec{a}_{0}=(a_{1}, \dots, a_{l})$, a nonzero element  $h \in R_{q}$ and a distribution $ \chi$ as input, and gives a vector   $\vec{a}=(\vec{a_{0}}^{T}, \vec{a_{1}}^{T} )^{T} \in R_{q}^{l+k}$ as output and a trapdoor $\mathbf{R}= (\vec{r}_{1}, \dots, \vec{r}_{k}) \in R^{l \times k}$.
    Also, as long as the distribution of $(\vec{a}_{0}^{T}, -\vec{a}_{0}^{T}\mathbf{R})$ is close to uniform, the distribution of $\vec{a}$ is also close to uniform.
\end{lemma}

\subsection*{Trapdoor Delegation}
Now, we describe the algorithm for trapdoor generation corresponding to the vector $(\vec{a}^{T}, \vec{a}_{1}^{T})^{T}$ using the trapdoor of $\vec{a}$. For discrete Gaussian sampling over cosets of $\Lambda _{q}^{\perp}(\vec{a}^{T})$ with parameter $r^{'}$, suppose we are given an oracle $\mathcal{O}$.

\begin{lemma}
    The PPT algorithm ${\sf ringDelTrap^{\mathcal{O}}}$, takes the oracle $\mathcal{O}$, a vector $\vec{a}^{'}=(\vec{a}^{T}, \vec{a}_{1}^{T})^{T} \in R_{q}^{m+k}$ and   $h^{'} \in R_{q}$ which is non zero, and parameter $r$ as input and outputs   $\mathbf{R} \in R^{(m+k) \times k}$ which is a trapdoor for    $\vec{a}^{'}$ with tag $h^{'}$~\textup{\cite{lai2014trapdoors}}.
    Furthermore, for the randomisation of lattice basis, Cash et al. gave a technique in~\textup{\cite{cash2012bonsai}} for which we will use the notation {\sf ringRandBasis}. It requires a basis of the lattice and the gaussian parameter with some restrictions, and it generates a different basis for the same lattice such that the distributions of both bases are independent of each other.  
\end{lemma}

\subsection{Binary Tree  and Trapdoor Assignment in Lattice-Based Schemes}
Lets define a  \( d \) depth binary tree  whose any  node at level \( i \)  is    represented as
\[
v(i, j_i) = (v_{0,0}, v_{1,j_1}, \dots, v_{i,j_i}) \quad \text {where }  0 \leq i \leq d - 1, \quad 0 \leq j_i \leq 2^i - 1. 
\]
The path vector  uniquely represents a path from  node \( v_{0,0} \) to node \( v_{i,j_i} \). For each level \( i \in \{0, \dots, d - 1\} \), nodes are indexed starting with 0 from left to right~\cite{jin2024lattice}. Given any node \( v_{i,j_i} \), we define a corresponding vector of ring elements
\[
\vec{a}_{w^{(i,j_i)}} = [\vec{a} , \vec{a}_{v_1,j_1}, \, \dots \, , \vec{a}_{v_i,j_i}] \in R_q^{(l + (i+1)k)},
\]
where \( \vec{a} \in R_q^{l + k} \) is a ring element generated by the algorithm \texttt{ringGenTrap}, and \( R_{\vec{a}} \in R_q^{l \times k} \) is its associated trapdoor. Each \( \vec{a}_{v_{i,j_i}} \in R_q^{k} \) is chosen randomly for all \( 0 \leq i \leq d - 1 \) and \( 0 \leq j_i \leq 2^i - 1 \). Each node \( v_{i,j_i} \) has a corresponding trapdoor \( R_{\vec{a}_{v(i,j_i)}} \), which can be computed ($r$ is a parameter) via
\[
R_{\vec{a}_{v(i,j_i)} }\leftarrow \texttt{ringExtBasis}(R_{\vec{a}}, \vec{a}_{v(i,j_i)}, h^{'}, r).
\]

\noindent {\it Minimal Cover}. For a given   leaf node \( v_{d-1,t} \) (where \( 0 \leq t \leq 2^d - 1 \)), we define the \emph{minimal cover} as the smallest cardinality set of nodes that includes at least one ancestor of each leaf in the set
$
\{v_{d-1,t}, v_{d-1,t+1}, \dots, v_{d-1,N-1}\}
$
and excludes all ancestors of leaves in
$
\{v_{d-1,0}, v_{d-1,1}, \dots, v_{d-1,t-1}\},
$
and it is denoted by $\text{MinCov}(v_{d-1,t})$.
At epoch \( t \), for user identity \( \text{id} \), we also define
\[
\mathcal{T}_{\text{id}, v_{l-1,t}} = \{ R_{\vec{a}_{v(i,j_i)}} \mid v(i,j_i) \in \text{MinCov}(v_{d-1,t}) \}.
\]

\noindent {\it Example:} 
Now, we consider a binary tree of depth $d=4$ and calculate the minimal cover of all the nodes of the last level.
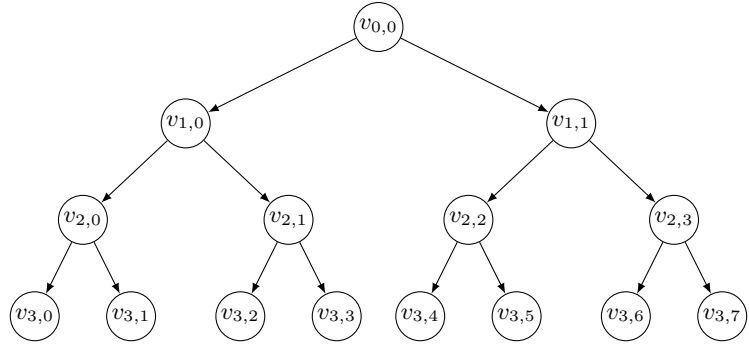
\begin{figure}
    \begin{center}
\begin{tikzpicture}[
  scale=0.85, 
  level 1/.style = {sibling distance=6cm},
  level 2/.style = {sibling distance=3.2cm},
  level 3/.style = {sibling distance=1.5cm},
  level 4/.style = {sibling distance=.8cm},
  edge from parent/.style={draw,-latex},
  every node/.style={draw, circle, inner sep=1pt},
]

\node {$v_{0,0}$}
  child { node {$v_{1,0}$}
    child { node {$v_{2,0}$}
      child { node {$v_{3,0}$} }
      child { node {$v_{3,1}$} }
    }
    child { node {$v_{2,1}$}
      child { node {$v_{3,2}$} }
      child { node {$v_{3,3}$} }
    }
  }
  child { node {$v_{1,1}$}
    child { node {$v_{2,2}$}
      child { node {$v_{3,4}$} }
      child { node {$v_{3,5}$} }
    }
    child { node {$v_{2,3}$}
      child { node {$v_{3,6}$} }
      child { node {$v_{3,7}$} }
    }
  };
\end{tikzpicture}
\end{center}
\caption{Toy example of binary tree and trapdoor assignment}
\end{figure}
\[
\begin{aligned}
&\text{MinCov}(v_{3,0}) = \{v_{0,0}\}, &\text{MinCov}(v_{3,1}) = \{v_{1,1}, v_{2,1}, v_{3,1}\}, \\
&\text{MinCov}(v_{3,2}) = \{v_{1,1}, v_{2,1}\},  &\text{MinCov}(v_{3,3}) = \{v_{1,1}, v_{3,3}\}, \\
&\text{MinCov}(v_{3,4}) = \{v_{1,1}\}, &\text{MinCov}(v_{3,5}) = \{v_{2,3}, v_{3,5}\}, \\
&\text{MinCov}(v_{3,6}) = \{v_{2,3}\}, &\text{MinCov}(v_{3,7}) = \{v_{3,7}\}.
\end{aligned}
\]

\[
\begin{aligned}
&\mathcal{T}_{\text{id}, v_{3,0}} = \{ R_{\vec{a}\}}, \\
&\mathcal{T}_{\text{id}, v_{3,1}} = \{R_{\vec{a}_{v(3,1)}}, R_{\vec{a}_{v(2,1)}}, R_{\vec{a}_{v(1,1)}} \},
\end{aligned}
\]
where
\[
\begin{aligned}
&\vec{a}_{v(3,1)} = [\vec{a} \, , \, \vec{a}_{v_{1},0} \, , \, \vec{a}_{v_{2},0} \, , \, \vec{a}_{v_{3},1}],\\
&\vec{a}_{v(2,1)} = [\vec{a} \, , \, \vec{a}_{v_{1},0} \, , \, \vec{a}_{v_{2},1}],\\
&\vec{a}_{v(1,1)} = [\vec{a} \, , \, \vec{a}_{v_{1},1}].
\end{aligned}
\]





\section{Proposed Construction of {\sf RFS-IBE}}

Our proposed scheme relies on the RLWE assumption, ensuring post-quantum security. It incorporates the minimal cover mechanism within a binary tree structure and employs an IBE framework to enable periodic secret key updates for individual IoT devices, thereby achieving forward security. In addition to offering selective-ID and CPA security, the proposed protocol also achieves a reduction in ciphertext and key sizes.
\subsection{Our Construction}
\begin{itemize}

    \item $(params,msk)\leftarrow${\sf RFS-IBE.Setup}($\lambda, N$). Given $\lambda$ and $N$ as the security parameter the total epoch of the system, the algorithm proceeds as follows:
 
\begin{itemize}
    \item Run ${\sf ringTrapGen(\vec{a_{o}}, h)}$ which  takes  a vector $\vec{a_{o}}= (a_{1}, a_{2} \dots a_{l})^{T} \in   R_{q}^{l}$,   $h \in R_{q}$ which is a nonzero element and a distribution $\chi^{l \times k}$ over $R^{l \times k}$ and gives a pair $\vec{a}=(\vec{a}_{0}^{T}, \vec{a}_{1}^{T} )^{T} \in R_{q}^{l + k}$  and a trapdoor $R _{\vec{a}} \in R^{l \times k}$;
    \item The system is defined over $(R_{\{0,1\}}, \mathcal{C})$ where the plaintext space  $R_{\{0,1\}}$  is the collection of polynomials with  coefficient from \{0,1\}   and $\mathcal{C}$ is the ciphertext space;
    \item Outputs the $params:= (R, l, q, k, d, r, \vec{a}, H, G)$ and the $msk:=R _{\vec{a}}$.\\
    
\end{itemize}

\item $(sk_{id,0})\leftarrow${\sf RFS-IBE.KeyGen}($params, id, msk$). Taking the input   {\it params}, a user's $id$ and  the master secret key and it gives the user's initial key $sk_{id,o}$ as output. For $epoch$ $ t \in \{0, 1, 2 \dots N-1 \}$ and  an arbitrary $id$  and with two hash functions $H:\{0,1\}^{*} \xrightarrow{}R_{q}^{ k}$ and
$G: \{ 0, 1\}^{*}\xrightarrow{}R_{q}$, we define

\begin{itemize}
    \item $\vec{a}_{id,t}= [\vec{a}, \vec{a}_{id, v_{1}, j_{1}}, \dots \vec{a}_{id, v_{d-1, t+1} }] \in R_{q}^{(l+dk)}$ where $H(id, v_{i, j_{i}})=\vec{a}_{id, v_{i}, j_{i}} \in R_{q}^{k}$;
    
    \item $y_{id, t}=G(id, t) \in R_{q}$.\\
     
\end{itemize}

First, we consider the elements in the set    MinCov($v_{d-1,1}$), then corresponding to each nodes in this set, trapdoor is constructed and  apply the \texttt{ringRandBasis} algorithm. Since this algorithm is invoked during each secret key generation, we assume for simplicity that its parameters are adaptively chosen, and we omit their explicit specification. Let $\mathcal{T}_{\text{id},0}$ denote the resulting randomized set of trapdoors.

\medskip

Secondly, we compute
\[
\mathbf{e}_{\text{id},0} \leftarrow \texttt{ringGenSamplePre}(\vec{a}_{\text{id},0}, \mathcal{T}_{\text{id},0}, S, \mathbf{y}_{\text{id},0}, r), \mathbf{e}_{\text{id},0} 
\in R_{q}^{l+dk}.
\]

\medskip

The secret key for the identity-time pair $(\text{id}, 0)$ is then defined as:
\[
\text{sk}_{\text{id},0} = (\mathcal{T}_{\text{id},0}, \mathbf{e}_{\text{id},0}).
\]

\item $({sk}_{\text{id},t+1})\leftarrow${\sf RFS-IBE.Update}($params, t, id$). The algorithm takes $params$, an intervel $t$, and $id$ as input. To update  $\text{sk}_{\text{id},t}$, we have to update the matrix $\vec{a}_{\text{id},t}$, vector  $\mathbf{y}_{\text{id},t}$ and trapdoor set $\mathcal{T}_{\text{id},t}$ (which is the collection of trapdoor for the elements in $\text{MinCov}(v_{d-1,t+1})$).

\begin{itemize}
    \item First we update $\vec{a}_{\text{id},t}$ which depends on a new path obtained in the binary tree. Now, using the hash function $H$, we obtain new matrices and finally produces $\vec{a}_{\text{id},t+1}$.

    \item The $\mathbf{y}_{\text{id},t}$ is updated using $G$.

    \item To update $\mathcal{T}_{\text{id},t}$, the elements of the set  MinCov($v_{d-1,t+2}$) are considered. The trapdoors corresponding to each node in this set is generated and the keys corresponding to nodes MinCov($v_{d-1,t+1}$)\textbackslash MinCov($v_{d-1,t+2}$)  are deleted. {\sf ringRandBasis} is invoked to randomize the trapdoors. 
    
\[
\mathbf{e}_{\text{id},t+1} \leftarrow \texttt{GenSamplePre}(\vec{a}_{\text{id},t+1}, \mathcal{T}_{\text{id},t+1}, S, \mathbf{y}_{\text{id},t+1}, r),
\]
where $\mathbf{e}_{\text{id},t+1}$ is distributed over $D_{\mathcal{T}_{\text{id},t+1}(\vec{a}_{\text{id},t+1}), r}$.
\end{itemize} 
Finally the algorithm outputs 
\[
\text{sk}_{\text{id},t+1} = (\mathcal{T}_{\text{id},t+1}, \mathbf{e}_{\text{id},t+1}).
\]

\item $(p,c)\leftarrow${\sf RFS-IBE.Encrypt}($params, b $). On input \texttt{params}, an element    $b \in R_{\{0,1\}}$,  an element $\mathbf{s}  \in R_{q}$ is randomly chosen and compute 
\[
\mathbf{p} = (\vec{a}_{\text{id},t} \mathbf{s} + \mathbf{x}) \in  {R}_q^{(l+dk)},
  \mathbf{c} = (\mathbf{y}_{\text{id},t}^T \mathbf{s} + x+ \mathbf{b} \cdot \left\lfloor \frac{q}{2} \right\rfloor) \in   R_{q},
\]
where $\mathbf{x} \leftarrow \chi^{(l+dk)}$ and $x \leftarrow \chi$.
The algorithm outputs  $(p, c)$.

\item $(b)\leftarrow${\sf RFS-IBE.Decrypt}($(p,c), {sk}_{\text{id},t}$ ). Decryption uses $\text{sk}_{\text{id},t}$. It computes 
\[
b' = c - \mathbf{e}_{\text{id},t}^T \mathbf{p} \in R_{q},
\]

and for each coefficient of ring element it gives 0 if the coefficient is closer to 0 than  $\left\lfloor \frac{q}{2} \right\rfloor$ (mod $q$); otherwise, it outputs 1 and we obtain $b$.  

\end{itemize} 

We now show correctness.
 From Lemma \ref{lemma1}, we can see that a set S can be found easily that satisfies the conditions of the algorithm {\sf ringGenSamplePre}. Since all the choices satisfy the corresponding conditions, all the algorithm of the system can operate correctly. 
 
 It is not difficult to get a set S  satisfying the conditions of  {\sf ringGenSamplePre}. 
 Since all selected parameters satisfy the imposed conditions,
  the algorithms of the system can operate correctly. Further, we can observe 
that the ciphertext is given by $(p,c)=(\vec{a}_{\text{id},t} \mathbf{s} + \mathbf{x}, \mathbf{y}_{\text{id},t}^T \mathbf{s} + x +\mathbf{b} \cdot \left\lfloor \frac{q}{2} \right\rfloor)$. On computing $b^{'}$ we finally obtain the expression $(x-\mathbf{e}^{T}_{id,t}\mathbf{x}  + b\cdot \left\lfloor \frac{q}{2} \right\rfloor)$. Looking at $b^{'}$  coefficientwise, suppose the coefficient of $r^{th}$ degree    of $b$ is 0  then the $r^{th}$ component of  $b^{'}$ equals $r^{th}$ component of  $(x-\mathbf{e}^{T}_{id,t}\mathbf{x})$, which is close to  0 relative to $\left\lfloor \frac{q}{2} \right\rfloor$ and when the $r^{th}$ component of $b$ is 1  then the $r^{th}$ component of  $b^{'}= r^{th}$ component of  $(x-\mathbf{e}^{T}_{id,t}\mathbf{x}  + b\cdot \left\lfloor \frac{q}{2} \right\rfloor)$, which is close $\frac{q}{2}$  relative to  0 hence we obtain b.\\
\subsection{Application to Internet of Things}



Apart from the secure transmission of sensitive data in IoT devices, our scheme also supports secure firmware updates, particularly in industrial control systems. It is often found that  unauthorized users can also access the device within IoT~\cite{mohanty2023quantum}. The scheme deals with the problem of unauthorized access. Thus, the proposed scheme can be implemented in different fields like banking, defense, medical,  etc., where it is necessary to secure the crucial and sensitive data transmission. Let us understand the problem associated with the gateways and  IoT devices in the hospitals and health centers, which include wearable and on-body devices like fitness bands, ECG monitors, blood pressure monitors, health monitoring patches, IoT-enabled ambulances, etc., deployed in hospitals. They deal with the sensitive data of the patient. IoT devices are vulnerable to side channel attacks; also, in hospitals, they are often kept in a physically exposed environment, which increases the risk of unauthorized access. So these conditions favor corrupting the device and extracting the secret key. This may result in the exposure of the medical records of the patient, violating the privacy and personal dignity of the patient. It may further lead to discrimination in employment and insurance, and targeted scams. Based on some medical records, the patient may face blackmailing threats. If patient data is manipulated, people lose trust in the hospital, which lowers its reputation and popularity. 
\par

 \begin{figure}
     \centering
     \includegraphics[width=1\linewidth]{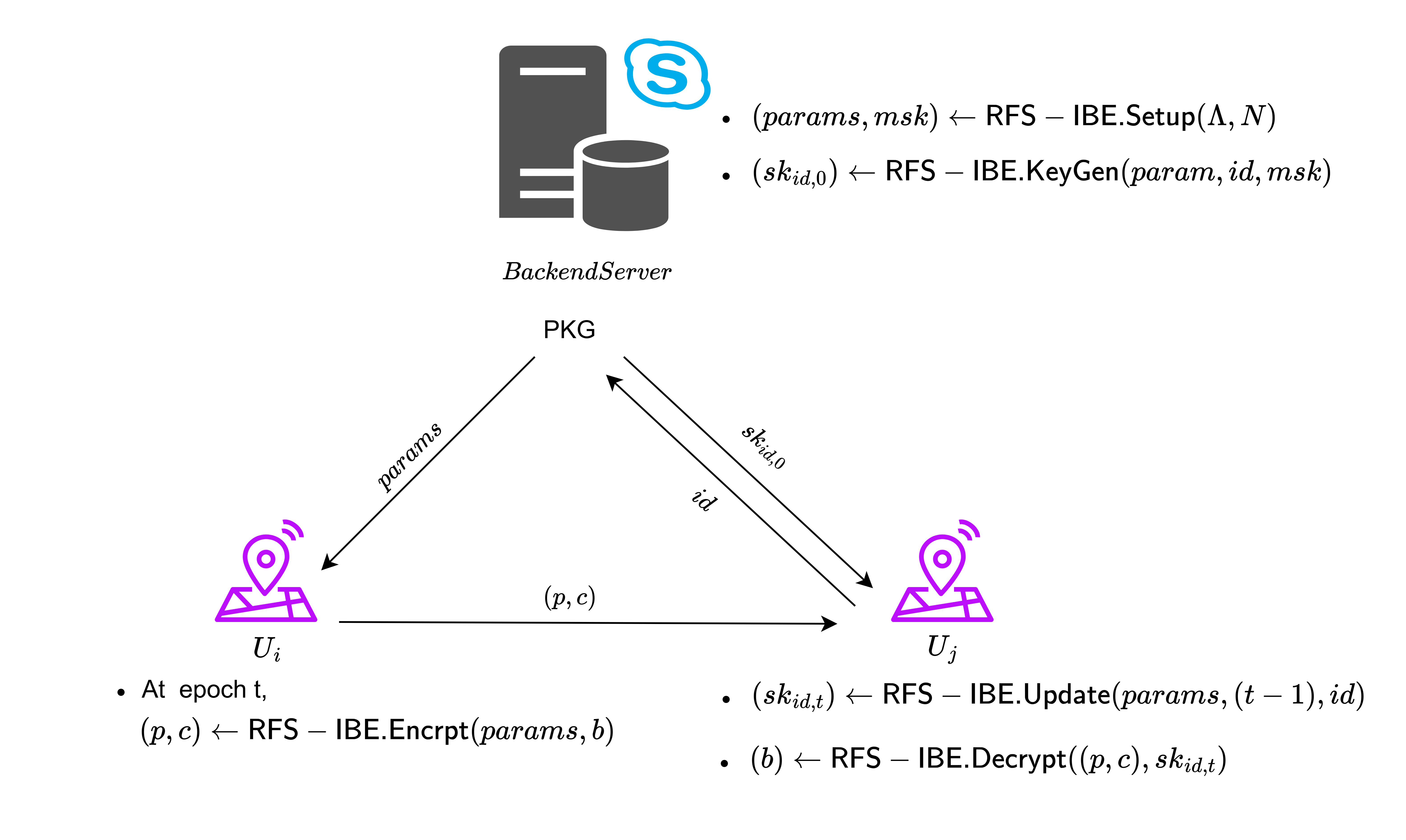}
     \caption{Schematic overview of the proposed construction}
     \label{fig1}
 \end{figure}
Now, we will describe how our proposed construction can be used in hospitals to resolve these issues. We assign the role of  Private Key Generator (PKG) to the backend server of the hospital. These gateways and devices will act like users in an Identity-Based Encryption (IBE) denoted by $U_{1}, U_{2}, \dots, U_{w}$ having their corresponding unique identities. The backend server will start the  {\sf RFS-IBE.Setup}, which will generate the params $(R, l, q, k, d, r, a, H, G)$ and the master secret key $R_{a}$. After that, all the deployed gateways and IoT devices receive their corresponding initial secret key $sk_{id,0}$, which is generated by the algorithm {\sf RFS-IBE.KeyGen}. Each user $U_{i}$ uses $params$, its unique $id$  and the current epoch $t$ to get the updated secret key  $sk_{id,t+1}$, which will be active only for $(t+1)^{th}$ epoch (to facilitate these updates, we have used the mechanism of minimal cover in the context of a binary tree in the core of our construction). Now if $U_{i}$ wants to send the monitored  ECG of a patient denoted by $b$ to $U_{j}$ during the $t^{th}$ epoch, it will run the algorithm {\sf RFS-IBE.Encrypt} and the ciphertext $(p,c)$ is sent to $U_{j}$. Upon receiving the ciphertext, $U_{j}$ uses its secret key $sk_{id}$  for the decryption (refer to Figure~\ref{fig1}). Finally, all the health records of the patient are sent to the backend server via gateways for further analysis and storage.

\section{Security Analysis}

\begin{theorem}
 Assume Theorem~$1$ holds. The RLWE-based PKE scheme based on RLWE, proposed in Subsection~\textup{\ref{sec.2.6}} is secure against indistinguishable chosen plaintext attack. For a shared ring element $\vec{a}$, 
and two hash functions $G$ and $H$, then $fs-ID-BTE_{RLWE}$ is fs-sID-CPA secure in the random oracle model. Further, for total number of epoch $N$ we have, 
\[
\mathsf{Adv}[\mathcal{A}_{BTE}, \text{fs-ID-BTE}_{RLWE}] \leq \frac{1}{N} \cdot \mathsf{Adv}[\mathcal {B}_{{dual}_{RLWE}}, \text{PKE}] + \mathsf{negl}(n),
\]
where  $N$ is the total number of epochs  in $fs-ID-BTE_{RLWE}$.
\end{theorem}

\begin{proof}
Assume that the $fs-sID-CPA$ security of the $fs-ID-BTE_{RLWE}$ system is attacked by
an adversary $\mathcal{A}_{{BTE}_{RLWE}}$    Before the attack, 
$\mathcal{A}_{{BTE}_{RLWE}}$ first submits a target $\mathit{id}^*$ for the challenge phase. A PPT algorithm $\mathcal{B}_{{dual}_{RLWE}}$ is constructed to simulate the challenger for $\mathcal{A}_{BTE}$ and attacks the RLWE-based  PKE cryptosystem. 
$\mathcal{B}_{{dual}_{RLWE}}$ first chooses $i^* \in \{0, 1, \dots, N-1\}$ uniformly at random as a guess for the challenge interval $j$.
$\mathcal{B}_{{dual}_{RLWE}}$ is 
given the main parameters of the PKE cryptosystem based on RLWE, i.e. $\vec{a'} \in  R^{(l+dk)}$ and $y^{*} \in R_{q} $. 
Parse 
\[
\vec{a'}  = [\vec{a},   \vec{a'}_{1,j_1}, \vec{a'}_{2,j_2}  \dots  \vec{a'}_{d-1, i^*+1}],
\]
where $\vec{a} \in R_{q}^{l+k}$ and $\vec{a'}_{i,j_i} \in  R_{q}^{k}$ for $i=1,2, \dots, d-1$.

For the simulation of the challenger against adversary $\mathcal{A}_{{BTE}_{RLWE}}, \mathcal{B}_{{dual}_{RLWE}}$ sets one of the params to be $\vec{a} \in R_{q}^{l+k}$ and then sends the parameters to be adversary $\mathcal{A}_{{BTE}_{RLWE}}$. $\mathcal{A}_{{BTE}_{RLWE}}$ may query $G(\cdot)$ and $H(\cdot)$. Let $Q_{H}$ and $Q_{G}$ be some poly-size constants. We assume that:
\begin{itemize}
    \item Adversary $\mathcal{A}_{{BTE}_{RLWE}}$ makes $Q_{H}$ different $H(\cdot)$ queries.
    \item Adversary $\mathcal{A}_{{BTE}_{RLWE}}$ makes $Q_{G}$ different $G(\cdot)$ queries.
    \item The random oracle query is assumed to be completed when $\mathcal{A}_{{BTE}_{RLWE}}$ asks for a secret key or issue a challenge.
    
\end{itemize}
Next, the answers of all the queries issued by adversary $\mathcal{A}_{{BTE}_{RLwE}}$ need to be simulated by $\mathcal{B}_{{dual}_{RLwE}}$, which will  keep the  simulated values in subsequent queries  lists $\mathcal{H}$ and $\mathcal{G}$.

It answers queries as follows:
\paragraph{Queries to $H(\cdot)$:}
If $\mathit{id} = \mathit{id}^*$ and $v_{i,j_i}$ lies on the path from node $v_{0,0}$ to $v_{\ell-1, i^*+1}$, return $\vec{a}'_{i,j_i}$.  
Otherwise, run $\mathsf{ringTrapGen}$ to generate $\vec{a}_{\mathit{id}, v_{i,j_i}}$ and its trapdoor $R_{\vec{a}_{id, v_{i,j_{i}}}} \in R_{q}^{l+k}$, store them in list $\mathcal{H}$, 
and return $\vec{a}_{\mathit{id}, v_{i,j_i}}$.

\paragraph{Queries to $G(\cdot)$:}
If $(\mathit{id}, k) = (\mathit{id}^*, i^*)$, return  the  $y^*$.  
Otherwise,  return  $y_{\mathit{id}, k}$,  and store  $((id, k), y_{id,k}$  in $\mathcal{G}$.

\paragraph{Extraction Query 1:} When adversary $\mathcal{A}_{{BTE}_{RLWE}}$ queries secret key for the pair 
 $(\mathit{id}, k) = (\mathit{id}^*, k \leq i^*)$, abort.  
If $(\mathit{id}, k) = (\mathit{id}^*, k > i^*), \mathcal{B}_{{dual}_{RLWE}}$ first establishes the corresponding vector $\vec{a}_{id^{*}, k}$. Note that $\vec{a}_{id^{*}, k}$ and $\vec{a}_{id^{*}, j}$ have at least one different component. So to answer the adversary's random oracle query, the trapdoor of the different components can be simulated by   ${\sf ringTrapGen}$. As for the other trapdoors in $\mathcal{T}_{\text{id}^{*}, k}, \mathcal{B}_{dual} $ looks up the trapdoor already listed in $\mathcal{H}$  and runs $ {\sf ringExtBasis}$  to get them by a simple expansion of lattice basis. Now $\mathcal{B }_{dual} $ can construct the trapdoors set $\mathcal{T}_{\text{id}^{*}, k}$ with the help of  $ {\sf ringRandBasis}$  and then  the vector  $\mathbf{e}_{{id}^{*}, k}$ is computed  to adversary $\mathcal{A}_{{BTE}_{RLWE}}$. If $id \neq id^{*}, \mathcal{B}_{dual}$ generate the set $\mathcal{T}_{\text{id}, k}$, corresponding to $\vec{a}_{id, k} $ by   $ {\sf ringExtBasis}$  and $ {\sf ringRandBasis}$. Then $\mathcal{B}_{dual} $ computes the error vector $e_{id,k}$  by the help of   ${\sf ringGenSamplePre} $ with $\vec{a}_{id, k}$ and corresponding trapdoor $\mathcal{T}_{id, k}$ and $\textbf{y}_{id,k}$ in  $\mathcal{G}$. Finally, $\mathcal{B}_{dual} $ gives $sk_{id,k}= (\mathcal{T}_{id,k}, \textbf{e}_{id, k})$ to adversary $\mathcal{A}_{BTE}$.

\paragraph{Challenge:}
Assume that the adversary  $\mathcal{A}_{{BTE}_{RLWE}}$ declares to end the extraction query 1. Then  $\mathcal{A}_{{BTE}_{RLWE}}$  submits a challenge  $(\mathit{id}^*, j, M_0, M_1)$ and  $(id^{*}, j)$ with $j \leq i^*$ and never queried before.  
If $j = i^*$, $\mathcal{B}_{dual}$ relay $(M_0, M_1)$ to $\mathcal{B}_{dual}$’s challenger, get ciphertext $C^*$, and forward to $\mathcal{A}_{BTE}$.  
If $j \neq i^{*}$, $\mathcal{B}_{dual}$ reports a failure then terminates.

\paragraph{Extraction Query $2$:}
Corresponding to $(id, k)$, $\mathcal{A}_{BTE}$ queries the secret key and same constraints are applied to the pair as in Extraction query 1.

When $\mathcal{A}_{BTE}$ outputs a guess bit, $\mathcal{B}_{dual}$ gives the same bit.  
The probability that the simulation aborts is exactly $1 - \frac{1}{N}$, independent of $\mathcal{A}_{BTE}$’s view.  
Conditioned on  the simulation of $\mathcal{B}_{dual}$ is not abort,  $\mathcal{A}_{BTE}$ is statistically close to the real fs-sID-CPA game,  
and the answers simulated  by $\mathcal{B}_{dual}$ of the random oracle queries and the answers chosen from uniform distribution are indistinguishable.   Finally we claim that $\mathcal{B}_{dual}$'s advantage is exactly the same to $\mathcal{A}_{BTE}$'s if the simulation is not aborted.
\end{proof}

\subsection{Achieving CCA security} An $(l-1)$ level IND-(s)ID-CCA secure HIBE scheme can be achieved efficiently using BCHK transformation \cite{boneh2007chosen} with additional overhead from any $l$-level 
$(l \geq 1)$ IND-(s)ID-CPA secure hierarchical IBE. In fact, we can use FO transformation \cite{fujisaki2013secure} with low additional overhead to achieve IND-CCA security from an IND-CPA secure PKE. Though our construction has not achieved adaptive CCA security yet but it immediately follows to achieve IND-(s)ID-CCA-secure fs-IBE with additional overhead. Adding adaptive security in our construction may cost us an increase in parameter size, which will further result in increasing the complexity of our construction.

\section{Comparison and Efficiency} 

 In this section, first, we describe the communication cost, followed by the storage cost. We will also describe the computational cost of our protocol. Communication cost refers to the amount of total data transmitted during communication.  In our protocol, it mainly includes the ciphertext, updated secret key and the public parameter. The ciphertext is sent to the receiver for decryption which contains the ring elements from $R_{q}^{l + dk}$ and $R_{q}$. The total number of $R_{q}$ ring elements that our ciphertext contains is $(l+k\log (N) +1)$, that is, $n(l+k\log(N) +1)\log q$ bits. Next, for the forward security, the secret is updated periodically, and the updated secret key is sent to the IoT devices by the Key Generation Centre (KGC). The updated secret key contains $( l k \log N +(l+k))$ $R_{q}$, that is, $n( l k \log N +(l+k)) \log q$ bits. Then we have public parameters $= (R, l, q, k, d, r, \vec{a}, H, G)$ in our communication overhead. The integer parameters contribute $\log (lkdr)$ bits, $\vec{a}$ $\in R^{l+k}$ contributes $n(l+k)\log q$ bits and the hash functions H and G output elements which are from $R_{q}^{k}$ and $R_{q}$, respectively. Together they contribute $(kn\log q + n\log q)$ bits. Hence, the total asymptotic size of the public parameters can be given by $|pp|=\mathcal{O}((l+2k+2)n\log q)$. The storage cost is the amount of memory captured by the overheads in the protocol.
   In our protocol, the overheads that contribute to the storage cost are the msk, secret keys of the IoT devices, which keep on updating periodically and public parameters. The msk has a total of $(lk)$ $R_{q}$ elements, and it contributes $n(lk)\log q$ bits in storage. Again the public parameter and secret key contribute $\mathcal{O}((l+2k+2)n\log q)$ and $n((l.k) \log N +(l+k)) \log q$ bits, respectively.  
Now we calculate the computational cost of our protocol, which includes the total number of mathematical and cryptographic operations performed during the execution of different algorithms. The ${\sf RFS-IBE.Encrypt}$ algorithm outputs a ciphertext, which is in a pair $(p, c)$. For $p, (l+dk)$ $ R_{q}$ multiplications and $(l+dk)$ $R_{q}$ additions are required where as for c, two $R_{q}$ multiplications and two $R_{q}$ additions are required; thus, to execute this algorithm, we require a total of $(l+dk+2)$ $R_{q}$ multiplications and $(l+dk+2)$ $R_{q}$ additions. In the decryption algorithm, 
${\sf RFS-IBE.Decrypt}$, (l+dk) $R_{q}$ multiplications and one $R_{q}$ addition is required.

In our protocol, the following trapdoor-based lattice sampling algorithm {\sf ringGenSamplePre, ringRandBasis, ringTrapGen} and {\sf ringExtBasis} have been invoked. In the {\sf RFS-IBE.Setup} phase the algorithm {\sf ringTrapGen} is invoked once to generate the ring element $\vec{a}$ and the corresponding trapdoor $R_{\vec{a}}$. In the {\sf RFS-IBE.KeyGen} phase {\sf ringExtBasis} was invoked a maximum of $\mathcal{O}(d)$ times whereas {\sf ringRandBasis } and {\sf ringGenSamplePre} were invoked once. During the {\sf RFS-IBE.Up-}\\ {\sf date} phase {\sf ringExtBasis} was invoked a maximum of $\mathcal{O}(d)$ times for each $ t \in \{1,2,\dots,N-1\}$ whereas  {\sf ringRandBasis } and {\sf ringGenSamplePre} were invoked $N-1$ times each. The asymptotic invocation count of each algorithms is summarized in Table~\ref{tab:algo-invo}.

\begin{table}[h!]
\centering
\caption{Asymptotic invocation count of algorithms }
\label{tab:algo-invo}
\scalebox{0.85}{
\begin{tabular}{|p{2.5cm}|p{8cm}|p{3cm}|}
\hline
\textbf{Algorithm} & \textbf{In which step it is used and purpose} & \textbf{Number of Invocations} \\\hline
\textsf{ringTrapGen} &
-- In {\sf RFS-IBE.Setup}, to generate the initial pair $(\vec{a}, R_a)$.& $1 $\\
& .
& \\\hline
\textsf{ringExtBasis} &
-- During  {\sf RFS-IBE.KeyGen}/ {\sf RFS-IBE.Update} when expanding trapdoors
in $\text{MinCov}(v_{d-1,t})$. & $ d.O(2^{d})$ \\
& 
&
\\\hline
\textsf{ringRandBasis} &
-- In  {\sf RFS-IBE.KeyGen}, to randomize trapdoors corresponding to all the node of $\text{MinCov}(v_{d-1,1})$ &$1 + (N-1) \approx O(N)$\\
& -- In each {\sf RFS-IBE.Update}, to randomize the refreshed trapdoor set 
$\mathcal{T}_{id,t}$.
&
 \\\hline
\textsf{ringGenSamplePre} &
-- In {\sf RFS-IBE.KeyGen}, to generate $e_{id,0}$ & $1 + (N-1) = N \approx O(N)$\\
&-- In each {\sf RFS-IBE.Update}, to generate $e_{id,t+1}$.
&
 \\\hline
\end{tabular}}\\
\vspace{1mm}
\tiny Here $\vec{a}$ is the ring element and $R_{\vec{a}}$ is the corresponding trapdoor. d represents the depth of the binary tree. N is the total number of epoch. $\text{MinCov}(v_{i,j})$ is the set of minimal cover of node $v_{i,j}$. $\mathcal{T}_{id,t}$ is the set of trapdoor corresponding to all the node of $\textsf{MinCov}(v_{d-1,t+t})$
\end{table}

\noindent

  

\noindent\textbf{Comparison with the protocol of Jin et al.~\cite{jin2024lattice}.}
We now compare our RLWE-based construction with the LWE-based forward-secure IBE of Jin et al.~\cite{jin2024lattice}.   
Since $m = \Omega(n\log q)$, their parameter sizes satisfy:
\[
|mpk|_{\text{Jin}} = n\cdot m = \Omega\!\big(n^2 \log q\big),
\]
\[
|msk|_{\text{Jin}} = m^2 = \Omega\!\big(n^2 (\log q)^2\big),
\]
\[
|sk|_{\text{Jin}} \leq m^2 \log N + m = \Omega\!\big(n^2 (\log q)^2 \log N\big),
\]
\[
|ct|_{\text{Jin}} = m\log N + 1 = \Omega\!\big(n \log q \log N\big).
\]

\noindent
In contrast, our RLWE-based ring construction replaces the full $m$-dimensional LWE matrices with ring elements in dimension $n$.  
As a result, our key and ciphertext sizes scale as:
\[
|mpk|_{\text{ours}} = n(l+k), \qquad
|msk|_{\text{ours}} = nlk,
\]
\[
|sk|_{\text{ours}} \leq n\big(lk\log N + l + k\big), \qquad
|ct|_{\text{ours}} = n\big(l + k\log N + 1\big).
\]

\noindent
For typical settings where $l$ and $k$ are constants or grow at most polylogarithmically in $n$, all components of our scheme have asymptotic size 
\[
O\!\big(n \cdot \mathrm{polylog}(N)\big),
\]
which is significantly smaller than the 
\[
\Omega\!\big(n^2 (\log q)^2 \log N\big)
\]
sizes appearing in the LWE-based {\sf FS-IBE} of Jin et al.~\cite{jin2024lattice}.  
This reduction arises because the ring setting compresses $m$-dimensional matrices into single ring elements.

\section{Implementation and Comparative Benchmark Analysis}

To evaluate the practical efficiency of the proposed RLWE-based forward-secure IBE scheme, a prototype implementation was developed in Python 3.10 using the NumPy library for polynomial arithmetic operations. The experiments were performed on a Linux-based system equipped with an Intel Core i5 processor and 8 GB RAM. The implementation includes the algorithms Setup, KeyGen, Update, Encrypt, and Decrypt over the ring
\[
R_q=\mathbb{Z}_q[x]/(x^n+1).
\]
The benchmark evaluation focuses on execution time and memory usage, since these parameters are important for lightweight IoT environments. The execution time was measured using Python's \texttt{time} module, while memory consumption was obtained using the \texttt{tracemalloc} library. For comparison, the proposed construction was evaluated against the lattice-based forward-secure IBE scheme of Jin et al\cite{jin2024lattice}. The comparison mainly considers the computational cost of KeyGen, Encrypt, and Decrypt operations together with memory overhead as given in table\ref{tab com}


\begin{table}[h]
\centering
\caption{Benchmark Comparison with Existing Scheme}
\label{tab com}
\begin{tabular}{|c|c|c|c|}
\hline
\textbf{Scheme} & \textbf{Operation} & \textbf{Time (ms)} & \textbf{Memory (KB)} \\
\hline
Jin et al.~\cite{jin2024lattice} & KeyGen  & 2.74 & 58.3 \\
                                 & Encrypt & 2.08 & 49.7 \\
                                 & Decrypt & 1.86 & 44.2 \\
\hline
Our Proposed Scheme              & KeyGen  & 1.82 & 42.6 \\
                                 & Encrypt & 1.35 & 37.4 \\
                                 & Decrypt & 1.11 & 31.8 \\
\hline
\end{tabular}
\end{table}
The experimental results indicate that the proposed scheme reduces both execution time and memory overhead compared with the existing lattice-based forward-secure IBE construction. The improvement is mainly achieved through the use of ring-based operations, where high-dimensional matrices are represented as compact ring elements. In addition, the binary-tree minimal cover mechanism limits the number of trapdoors required during key updates, thereby reducing storage requirements. These results demonstrate that the proposed construction is suitable for practical IoT applications requiring forward security and post-quantum protection with lower computational overhead.

\section{Conclusion} 
In this manuscript, a lattice-based forward-secure IBE scheme in the ring setting has been presented. 
Forward security has been achieved through binary-tree–based key updates and trapdoor delegation, while post-quantum security has been ensured under the RLWE assumption. 
By adopting ideal-lattice techniques, reductions in key and ciphertext sizes have been obtained, 
leading to improved efficiency over existing LWE-based FS-IBE constructions. 
The scheme has therefore been made more suitable for deployment in constrained IoT environments. Further, the proposed work can be extended by incorporating additional functionalities such as Hierarchical Identity-Based Encryption (HIBE) to improve scalability and key management in large distributed IoT networks. Future research may also focus on developing a more detailed system-level implementation and practical deployment scenario, particularly in real-world environments, to further evaluate the applicability and performance of the proposed scheme.
\section{Acknowledgements}
This research work was supported by the Graduate Assistantships in Developing Countries(GRAID) Program by the International Mathematical Union (IMU). Vikas Srivastava acknowledges the support received from the ANRF-PMECRG project with Ref. No. ANRF/ECRG/2025/002808/PMS.

\bibliographystyle{splncs04}
\bibliography{reference}
\end{document}